\documentclass[conference]{IEEEtran}
\IEEEoverridecommandlockouts
\usepackage{cite}
\usepackage{amsmath,amssymb,amsfonts}
\usepackage{algorithmic}
\usepackage{graphicx}
\usepackage{textcomp}
\usepackage{xcolor}
\usepackage{graphicx}
\usepackage{amsthm}
\usepackage{empheq}
\usepackage{bm}
\usepackage{booktabs}
\usepackage{pagecolor}
\usepackage{caption}
\usepackage{subcaption}
\usepackage{tikz,pgfplots}
\usetikzlibrary{matrix,decorations.pathreplacing,positioning}
\usepackage{enumitem}
\usepackage{wasysym}

\usepackage[margin=3cm]{geometry}

\numberwithin{equation}{section}

\theoremstyle{definition}
\newtheorem{theorem}{Theorem}[section]

\newtheorem{proposition}[theorem]{Proposition}
\newtheorem{definition}[theorem]{Definition}
\newtheorem{example}[theorem]{Example}

\newtheorem{notation}[theorem]{Notation}
\newtheorem{remark}[theorem]{Remark}

\newtheorem{lemma}[theorem]{Lemma}

\newcommand{\F}{\mathbb{F}}
\newcommand{\R}{\mathbb{R}}
\newcommand{\mN}{\mathcal{N}}
\newcommand{\alp}{\mA}
\newcommand{\ad}{\boldsymbol{A}}
\newcommand{\U}{\mathcal{U}}
\newcommand{\mF}{\mathcal{F}}
\newcommand{\mV}{\mathcal{V}}
\newcommand{\mE}{\mathcal{E}}

\newcommand{\C}{\textup{C}}
\newcommand{\Fq}{\F_q}

\newcommand{\bT}{\boldsymbol{T}}

\newcommand{\mU}{\mathcal{U}}
\newcommand{\mA}{\mathcal{A}}
\newcommand{\mD}{\mathcal{D}}
\newcommand{\mS}{\mathcal{S}}
\newcommand{\h}{\textup{H}}
\newcommand{\concat}{\RHD} 

\newlength{\mynodespace}
\definecolor{myg}{RGB}{220,220,220}
\usepackage{soul}
\pgfplotsset{compat=1.18}
\usepackage{calrsfs}
\usepackage{float}
\usepackage{soul}

\usepackage{hyperref}
\def\BibTeX{{\rm B\kern-.05em{\sc i\kern-.025em b}\kern-.08em
    T\kern-.1667em\lower.7ex\hbox{E}\kern-.125emX}}
\begin{document}

\title{Multishot Adversarial Network Decoding\\
\thanks{The work of G. C. and G. L. M. is supported by NSF DMS-2201075. G. L. M. is also supported in part by the Commonwealth Cyber
Initiative. A. R. is supported by the Dutch Research Council through grants VI.Vidi.203.045, OCENW.KLEIN.539, 
and by the Royal Academy of Arts and Sciences of the Netherlands.}
}

\author{\IEEEauthorblockN{Giuseppe Cotardo, Gretchen L. Matthews, Julia Shapiro}
\IEEEauthorblockA{\textit{Dept. of Mathematics} \\
\textit{Virginia Tech}\\
Blacksburg, VA, USA \\
gcotardo@vt.edu, gmatthews@vt.edu, juliams22@vt.edu}
\and
\IEEEauthorblockN{Alberto Ravagnani}
\IEEEauthorblockA{\textit{Dept. of Mathematics and Computer Science} \\
\textit{Eindhoven University of Technology}\\
Eindhoven, the Netherlands  \\
a.ravagnani@tue.nl}
}

\maketitle

\begin{abstract}
We investigate adversarial network coding and decoding focusing on the multishot regime.
Errors can occur on a proper subset of the network edges and are modeled via an adversarial channel. 
The paper contains both bounds and capacity-achieving schemes
for the Diamond Network and the Mirrored Diamond Network. We also initiate the study of the generalizations of these networks.

\end{abstract}


\section{Introduction}

Network coding is a communication strategy that outperforms routing \cite{yueng2006upper, linear2003network, network2000flow, ad2005network, medard2003, wang2007broadcast,byzantine2007,random2008network,nutman2008,kotter2008,jafari2009multi,kschischang2019multi,multi2009bound}. 
In that context, sources attempt to transmit information packets to one or more terminals through a network of intermediate nodes.
In~\cite{ravagnani2018}, the authors introduced the problem of computing the capacity of a network where errors 
can occur on a proper subset of the network edges. This scenario was modeled via an adversarial channel and a generalized Singleton Bound for the one-shot capacity (i.e., the largest amount of information packets that can be sent to all terminals in a single use of the network) of networks with restricted adversaries was established.

This study was furthered in~\cite{beemer2021curious, BEEMER202236,beemer2023network}, where 
the authors introduced the concept of \textit{network decoding} as a necessary strategy to achieve capacity in networks with a restricted adversary. In~\cite{beemer2023network}
five elementary families of
networks were introduced, together with upper bounds on their one-shot capacity and communication schemes.
The Diamond Network and the Mirrored Diamond Network are 
examples from these elementary families.
The results of~\cite{beemer2023network}
demonstrate that when an adversary is confined to operate on a vulnerable region of the network, end-to-end communication strategies combined with linear network coding are sub-optimal.
 
In this paper, we initiate the study of the \textit{multishot capacity} of networks with restricted adversaries. In other words, we wish to compute the largest amount of information packets that can be sent to all terminals on average over multiple uses of the network. We focus on the Diamond Network and the Mirrored Diamond Network as building blocks of a more general theory that will be presented in an extended version of this work.

Our work gives insight into when using a network multiple times for communication allows us to send more information than in one use of a network with a large alphabet. We show that the multishot capacity of the Diamond Network has a strict increase in comparison to its one-shot capacity in one adversarial model. On the other hand, for the Mirrored Diamond Network and generalizations, the maximum capacity is the same over multiple uses of these networks. 

This paper is organized as follows. In Section~\ref{sec:netcod} we introduce the notation and background to be used throughout the paper. Section~\ref{sec:regI} establishes the exact multishot capacity of the Diamond Network, the Mirrored Diamond Network and the families of networks introduced in \cite{BEEMER202236, beemer2023network}, under the assumption that the adversary cannot change the edges attacked in each transmission round. In Section~\ref{sec:regII}, we study the multishot capacities of these networks assuming that the adversary has freedom to change the edges attacked in each transmission round. In Section~\ref{sec:future}, we include open questions and ideas for future work. 

\section{Network Decoding}
\label{sec:netcod}
This section includes preliminary definitions and results related to adversarial networks. We establish the notation that will be used throughout the paper.

In the following, we let $q$ be a prime power and $\Fq$ the finite field with $q$ elements. We recall the definition of \textbf{network} proposed in~\cite[Definition 38]{ravagnani2018}.  A \textbf{network} is a 4-tuple $\mN = (\mV,\mE,S,\bT)$ where $(\mV,\mE)$ is a directed, acyclic and finite multigraph, $S \in \mV$ is the \textbf{source} and  $\bT \subset \mV$ is the set of \textbf{terminals}. 
We also assume that there exists a directed path from $S$ to any $T \in \textbf{T}$, $|\textbf{T}| \geq 1$ and $S \notin \textbf{T}$ and that for every $V \in \mathcal{V}\setminus (\{S\}\cup \textbf{T}),$ there exists a directed path from $S$ to $V$ and from $V$ to one of the terminals $T \in \textbf{T}$.
The elements of $\mathcal{V}$ are called \textbf{vertices} (or \textbf{nodes}) and the ones in $\mathcal{E}$ are called \textbf{edges}. The refer to the elements in $V \in \mathcal{V}\setminus (\{S\}\cup \textbf{T})$ as \textbf{intermediate vertices} (or \textbf{intermediate nodes}). We define \textbf{indegree} (respectively, \textbf{outdegree}) of a vertex $V\in\mathcal{V}$ to be the the set of edges coming into (respectively, going out of) $V$. We denote it by $\textup{in}(V)$ (respectively, $\textup{out}(V)$) and its cardinality by $\deg^{+}(V)$ (respectively, $\deg^{-}(V)$). 

Each edge of the network $\mN$ carries one element from an \textbf{alphabet} $\mA$. We assume that $\alp$ has cardinality at least $2$. The vertices in the network $\mN$ receive symbols from $\mA$ over the incoming edges, process them according to functions, and then send the outputs over the outgoing edges. We model errors in the transmission as being introduced by an adversary $A$ who can corrupt up to $t$ edges from a fixed subset $\mathcal{U} \subseteq \mathcal{E}$. The adversary can change a symbol sent across one of the edges of $\mathcal{U}$ and change it to any other symbol of the alphabet. We call the pair $(\mN,\textbf{A})$ an \textbf{adversarial network}. 

\begin{definition}[\hspace{-0.1pt}{\cite[Definition 1]{ravagnani2018}}]
 An \textbf{adversarial channel} is a map $\Omega := \mathcal{X} \to 2^{\mathcal{Y}}$, where $\mathcal{X}$ and $\mathcal{Y}$ are the \textbf{input} and \textbf{output} alphabets respectively and $\mathcal{X},\mathcal{Y} \neq \emptyset$. The notation for this type of channel is $\Omega: \mathcal{X} \dashrightarrow \mathcal{Y}$.
 \end{definition}

\begin{definition} 
[\hspace{-0.1pt}{\cite[Definition~40]{ravagnani2018}}]
A \textbf{network code} $\mathcal{F}$ for $\mN$ is a family of functions $\{\mathcal{F}_V: V \in 
  \mathcal{V} \setminus\{S \cup \textbf{T}\}\},$ where $\mathcal{F}_{V}: \mA^{\textup{deg}^{+}(V)} \to \mA^{\textup{deg}^{-}(V)}$ for all $V$.
\end{definition}

The functions in $\mF$ give a description of how $\mN$ processes the information in each node coming from the incoming edges. This can be uniquely interpreted according to a partial ordering. Let $e_1$, $e_2$ be edges in a network. We say $e_1 \in \mE$ \textbf{precedes} $e_2 \in \mE$ ($e_1 \preceq e_2$) if there exists a directed path in $(V,\mE)$, with starting edge $e_1$ and ending edge $e_2$. It is well know that this partial order can be extended to a total order on $\mE$, denoted by $\leq$. Let $\mU,\mU' \in \mE$ be subsets of edges in a network. As in \cite[Definition 3]{BEEMER202236}, we say that~$\mU$ \textbf{precedes} $\mU'$ if every path from $S$ to an edge of $\mU'$ contains an edge of $\mU$.

\begin{definition}[\hspace{-0.1pt}{\cite[Definition 4]{BEEMER202236}}]\label{notC}
Let $(\mN,\ad)$ be an adversarial network 
and let $\U, \U'\subseteq \mE$ be non empty with $\U$ preceding $\U'$. Let $\mF$ be a network code for $\mN$. For $x \in \alp^{|\U|}$, we denote the set of vectors over the alphabet that can be exiting the edges of $\U'$ as 
\[ \Omega[\mN, \ad, \mF, \U \to \U' ](x) \subseteq \alp^{|\U'|}\]
We note that the vertices process the information according to  the choice of $\mF$, the coordinates of $x$ are values from $\alp$ entering the edges of $\U'$ and we consider the total order $\leq$ throughout.
\end{definition}

\begin{definition}[\hspace{-0.1pt}{\cite[Definition 6]{BEEMER202236}}]
An \textbf{outer} \textbf{code} for 
$\mN$ is a subset
$C \subseteq \alp^{\deg^{+}(V)}$ with $|C| \ge 1$. We say that $C$ is  \textbf{unambiguous} (or \textbf{good}) for $(\mN, \ad, \mF)$ if for all $x,y \in C$ with $x \neq y$ and for all $T \in \bT$ we have
\begin{multline*}
    \Omega[\mN, \ad, \mF, \textup{out}(V) \to \textup{in}(V) ](x)\cap \\
     \Omega[\mN, \ad, \mF, \textup{out}(V) \to \textup{in}(V) ](y) =\emptyset.
\end{multline*}
\end{definition}

We note that the intersection being empty guarantees that every element of $C$ can be uniquely recovered by every terminal, despite the actions the adversary takes. Next we define the notion of one-shot capacity of an adversarial network.

\begin{definition}[\hspace{-0.1pt}{\cite[Definition~3.18]{beemer2023network}}] The \textbf{one-shot} capacity of an adversarial network $(\mathcal{N},\ad)$ is the maximum $\alpha \in \R$ such that there exists an unambiguous code $C$ and a network code $\mathcal{F}$ with $\alpha = \log_{|\mathcal{A}|}(|C|).$ We denote the maximum quantity by $C_1(\mathcal{N},\ad)$.

\end{definition}

The following result, proved in~\cite{ravagnani2018}, gives an upper bound on the one-shot capacity of a network $\mathcal{N}$ under the presence of an adversary who is restricted to corrupting a proper subset of its edges.

\begin{theorem}[The Singleton Cut-Set Bound - {\hspace{-0.1pt}\cite[Theorem 67]{ravagnani2018}}]
\label{cutset}
Let $t\geq 0$ and suppose that $\ad$ can corrupt up to $t$ edges from a subset $\mathcal{U} \subseteq \mathcal{E}$. Then 
\begin{multline*}
C_1(\mN,\ad) \leq \min_{T \in \textbf{T}}\min_{\mathcal{E}'} \left( |\mE' \setminus \mathcal{U}|\right. \\\left.+\max\{0, |\mE' \cap \mathcal{U}|-2t\} \right)
\end{multline*}
where $\mE' \subseteq \mE$ ranges over all edge-cuts between $S$ and $T$.
\end{theorem}

This work focuses on the \textbf{multishot capacity} of an adversarial network $(\mN,\ad)$. A formal definition, which extends~\cite[Definition~3.18]{beemer2023network}, is the following.
\begin{definition}
    Let $i$ be a positive integer. The \textbf{i-th shot capacity} of $(\mN,\ad)$ is the maximum $\alpha\in\R$ such that there exists an unambiguous code $C$ for $(\mN,\ad)$ with $\alpha=\frac{\log_{|\mA|}(|C|)}{i}$. We denote this maximum value by $\C_i(\mN,\ad)$.
\end{definition} 

Using the network multiple times can also be modeled as the \textit{power channel} associated to it. We can formalize it as follows.

\begin{definition}[\hspace{-0.1pt}{\cite[Definition~10]{ravagnani2018}}]
\label{power}Let $i \ge 1$ be an integer. The \textbf{i-th power} of a channel 
$\Omega : \mathcal{X} \dashrightarrow \mathcal{Y}$ is represented by the channel
$$\Omega^i := \underbrace{\Omega \times \cdots \times 
\Omega}_{i \text{ times}} :
\mathcal{X}^i \dashrightarrow \mathcal{Y}^i.$$
\end{definition} 

We also recall the definitions the operations of \textit{product} and \textit{concatenation}  of channels introduced in~\cite{ravagnani2018}. See~\cite[Definitions~7 and~14]{ravagnani2018} respectively.

Let $\Omega_1:\mathcal{X}_1 \dashrightarrow \mathcal{Y}_1$ and $\Omega_2:\mathcal{X}_2 \dashrightarrow \mathcal{Y}_2$
be channels and assume that $\mathcal{Y}_1 \subseteq \mathcal{X}_2$. The \textbf{product} of $\Omega_1$ and $\Omega_2$
is the channel $\Omega_1 \times \Omega_2 : \mathcal{X}_1 \times \mathcal{X}_2 \dashrightarrow  \mathcal{Y}_1 \times \mathcal{Y}_2$
defined by   
$$(\Omega_1 \times \Omega_2)(x_1,x_2):= \Omega_1(x_1) \times 
\Omega_2(x_2),$$ $\mbox{for all 
$(x_1,x_2) \in \mathcal{X}_1 \times \mathcal{X}_2$}$. The \textbf{concatenation} of
$\Omega_1$ and $\Omega_2$ is the channel
$\Omega_1 \blacktriangleright \Omega_2 : \mathcal{X}_1 \dashrightarrow \mathcal{Y}_2$ defined by
$$(\Omega_1 \blacktriangleright \Omega_2)(x):= \bigcup_{y \in \Omega_1(x)} \Omega_2(y).$$

The following proposition provides a lower bound on the one-shot capacity of product channels and establishes a connection between the one-shot capacity and the $i-$th shot capacity of a network.

\begin{proposition}[\hspace{-0.1pt}{\cite[Proposition~12]{ravagnani2018}}]
\label{lwpord}
For a channel $\Omega: \mathcal{X} \dashrightarrow \mathcal{Y}$ and any $i\geq 1$, we have
$$C_1(\Omega^i) \ge i \cdot C_1(\Omega).$$ 
\end{proposition}

We restrict to networks with a single source $S$ and we follow the notation introduced in~\cite[Section~V]{beemer2023network}. A network $\mN = (\mV,\mE,S,\bT)$ is \textbf{simple} if it has only one terminal $T$, i.e. $\bT=\{T\}$. We say that $\mN$ is a \textbf{simple two-level network} if every path from $S$ to $T$ has length $2$. 
Let $\mN$ be a simple two-level network with $j\geq 1$ intermediate nodes. In the sequel we follow~\cite[Notations~5.3 and Example~5.4]{beemer2023network}, and refer the reader to that paper for the details. Let $[x_1,\ldots,x_j]$ and $[y_1,\ldots,y_j]^\top$ be the matrix representation of the graph induced by the source, the intermediate nodes, and the terminal respectively. We denote $\mN$ by $([x_1,\ldots,x_j],[y_1,\ldots,y_j])$.

In the remainder of the paper, we will focus on particular families of networks introduced in~\cite{beemer2021curious}, \cite{BEEMER202236} and~\cite{beemer2023network}. We recall their definitions.

Let $\mD$ be the network in Figure~\ref{diamond} and consider an adversary $\ad_{\mD}$ able to corrupt at most one of the dashed edges. We call the pair $(\mD,\ad_{\mD})$ the \textbf{Diamond Network}. It was shown in~\cite[Section~III]{BEEMER202236} and \cite{beemer2021curious} that this is the smallest example of a network that does not meet the Singleton Cut-Set Bound proved in \cite{ravagnani2018}. In particular, the following holds.
\begin{theorem}[\hspace{-0.1pt}{\cite[Theorem 13]{BEEMER202236}}] For any alphabet $\mA$, we have $ C_1(\mathcal{D}, \ad_{\mathcal{D}}) = \textup{log}_{|\alp|}(|\alp|-1).$
\end{theorem}

\begin{figure}[ht]
\centering
\begin{tikzpicture}
\tikzset{vertex/.style = {shape=circle,draw,inner sep=0pt,minimum size=2.0em}}
\tikzset{nnode/.style = {shape=circle,fill=myg,draw,inner sep=0pt,minimum
size=2.0em}}
\tikzset{edge/.style = {->,> = stealth}}
\tikzset{ddedge/.style = {dashed,->,> = stealth}}

\node[vertex] (S) {$S$};

\node[shape=coordinate,right=\mynodespace of S] (L) {};

\node[nnode,above=0.5\mynodespace of L] (V1) {$V_1$};

\node[nnode,below=0.5\mynodespace of L] (V2) {$V_2$};

\node[vertex,right=\mynodespace of L] (T) {$T$};

\draw[ddedge,bend left=0] (S)  to node[sloped,fill=white, inner sep=1pt]{\small $e_1$} (V1);

\draw[ddedge,bend left=16] (S) to  node[sloped,fill=white, inner sep=1pt]{\small $e_2$} (V2);

\draw[ddedge,bend right=16] (S)  to node[sloped,fill=white, inner sep=1pt]{\small $e_3$} (V2);

\draw[edge,bend left=0] (V1)  to node[sloped,fill=white, inner sep=1pt]{\small $e_4$} (T);

\draw[edge,bend left=0] (V2)  to node[sloped,fill=white, inner sep=1pt]{\small $e_{5}$} (T);

\end{tikzpicture} 
\caption{\label{diamond}{{The Diamond Network $\mD$}}}
\end{figure}
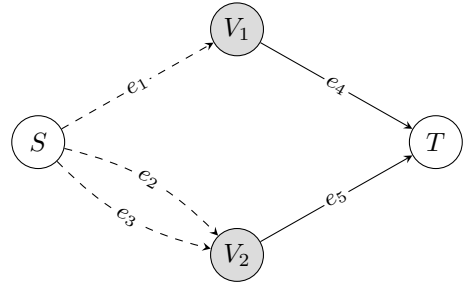

In~\cite[Section~3 and 4]{BEEMER202236}, it was proved that the network obtained by adding an edge to the Diamond Network, as in Figure~\ref{mirrored}, attains the Singleton Cut-Set Bound. The pair $(\mS,\ad_{\mS})$ is called \textbf{Mirrored Diamond Network}, where $\ad_{\mS}$  again represents an adversary able to corrupt at least on dashed edge.

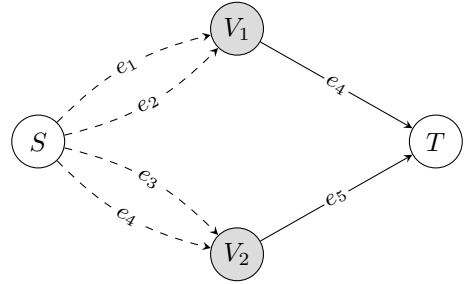
\begin{figure}[ht]
\centering
\begin{tikzpicture}
\tikzset{vertex/.style = {shape=circle,draw,inner sep=0pt,minimum size=2.0em}}
\tikzset{nnode/.style = {shape=circle,fill=myg,draw,inner sep=0pt,minimum
size=2.0em}}
\tikzset{edge/.style = {->,> = stealth}}
\tikzset{ddedge/.style = {dashed,->,> = stealth}}

\node[vertex] (S) {$S$};

\node[shape=coordinate,right=\mynodespace of S] (L) {};

\node[nnode,above=0.5\mynodespace of L] (V1) {$V_1$};

\node[nnode,below=0.5\mynodespace of L] (V2) {$V_2$};

\node[vertex,right=\mynodespace of L] (T) {$T$};

\draw[ddedge,bend left=16] (S)  to node[sloped,fill=white, inner sep=1pt]{\small $e_1$} (V1);

\draw[ddedge,bend right=16] (S) to  node[sloped,fill=white, inner sep=1pt]{\small $e_2$} (V1);

\draw[ddedge,bend left=16] (S)  to node[sloped,fill=white, inner sep=1pt]{\small $e_3$} (V2);

\draw[ddedge,bend right=16] (S)  to node[sloped,fill=white, inner sep=1pt]{\small $e_4$} (V2);

\draw[edge,bend left=0] (V1)  to node[sloped,fill=white, inner sep=1pt]{\small $e_4$} (T);

\draw[edge,bend left=0] (V2)  to node[sloped,fill=white, inner sep=1pt]{\small $e_{5}$} (T);

\end{tikzpicture} 
\caption{{{\label{mirrored} The Mirrored Diamond Network $\mS$}}}
\end{figure}

\begin{theorem}[\hspace{-0.1pt}{\cite[Proposition~14]{BEEMER202236}}] For any alphabet $\mA$, we have $C_1(\mathcal{S}, \ad_{\mathcal{S}}) = 1$.
\end{theorem}

The Diamond Networks were further generalized in~\cite[Section~V.C]{beemer2023network}. Let $t\geq 2$ be an integer and define the simple two-level networks $\mathfrak{C}_t = ([t,t+1],[t,t])$ and $\mathfrak{D}_t = ([2t,2t],[1,1])$. It is not hard to check that $\mathfrak{C}_1$ and $\mathfrak{D}_1$ recovers the Diamond and the Mirrored Diamond Network respectively.

We let $i \in \mathbb{N}$ represent the number of uses of each network. The remainder of the paper is organized according to the following adversarial models. 

\begin{enumerate}[label=A.\arabic*]
    \item\label{scenario1} The adversary corrupts the same $t$ edges over the multiple uses of the network.
    \item\label{scenario2} The adversary can change the $t$ edges to corrupt over the multiple uses of the network.
\end{enumerate}

\section{Multishot Regime I}
\label{sec:regI}
In this section, we study the multishot capacity of the Diamond and Mirrored Diamond Network introduced in~\cite{BEEMER202236} along with the families of networks introduced in \cite{beemer2023network}, in Scenario \ref{scenario1}.

\subsection{The Diamond Network}
We assume that $t=1$ in this section. We start by noticing that reusing the strategy previously proposed in \cite[Proposition 11]{BEEMER202236} one can easily show that $C_i(\mD, \ad_{\mD}) \geq  \textup{log}_{|\alp|}(|\alp|-1)$. We let $\Omega[\mathcal{D}, \ad_{\mathcal{D}}, \mF, \textup{out}(S) \to \textup{in}(T)]$ be the channel representing the transfer from $S$ to $T$ of the Diamond Network as in Definition \ref{notC} for the remainder of the paper. The aim of this section is to explicitly compute the multishot capacity of the Diamond Network in Scenario \ref{scenario1}. We provide a construction of an unambiguous code of cardinality $|\alp|^i - 1$ for $(\mD,\ad_\mD,\mF)$ that models $i$ transmission rounds and we show that this is the maximum possible in this setting. In particular, we prove that 
\begin{equation*}
\C_i(\mD,\ad_\mD)=\frac{\log_{|\mA|}(|\mA|^i-1)}{i}
\end{equation*}
We start with the following result. 

\begin{proposition}\label{prop:lowA1}
   Let $\mF$ be a network code for $(\mD,\ad_\mD)$. If $C\subseteq(\alp^{i})^3$ is an unambiguous code for $(\mD,\ad_\mD,\mF)$ then $|C|\geq |A|^i-1$.
\end{proposition}
\begin{proof}
    Let $\star$ be a symbol of the alphabet $\mA$. We want to show that the code 
    \begin{multline*}
        C=\{(a\mid a\mid a): a\in\alp\setminus(\star,\ldots,\star)\}\subseteq (\alp^{i})^3
    \end{multline*}
    is unambiguous for $(\mD,\ad_\mD)$. At each round, we use the same network code $\mF$ as in the proof of~\cite[Proposition 11]{BEEMER202236}. Suppose the adversary corrupts $e
    _1$ and recall that in this case the adversary is forced to change the symbol. One can check that, for any $a\in\alp^i$, we have
    \begin{equation*}
        \Omega^i[\mD, \ad, \mF, \textup{out}(S) \to \textup{in}(T) ]((a\mid a\mid a))=\{(b \mid a)\}
    \end{equation*}
    for some $b\in\alp^i\setminus\{(\star,\ldots,\star)\}$. On the other hand, if the adversary corrupts $e_2$ or $e_3$ then,  for any $a\in\alp^i\setminus\{(\star,\ldots,\star)\}$, we get
    \begin{multline*}
        \Omega^i[\mD, \ad, \mF, \textup{out}(S) \to \textup{in}(T)]((a\mid a\mid a))= \\
        \{(a\mid\star,\ldots,\star)\}.
    \end{multline*}
    It follows that, for any $c,c'\in C$, we have 
    \begin{multline*}
         \Omega^i[\mD, \ad, \mF, \textup{out}(S) \to \textup{in}(T) ](c)\cap\\
            \Omega[\mN, \ad, \mF, \textup{out}(S) \to \textup{in}(T) ](c')\neq\emptyset
    \end{multline*}
    if and only if $c= c'$ which implies that the code $C$ is unambiguous. This concludes the proof.
\end{proof}

Note that the appropriateness of the construction of $C$ in Proposition \ref{prop:lowA1} relies on the adversary not being able to change the edge attacked, meaning we know the exact location of the adversary for the next $i-1$ transmission rounds. Therefore, the strategy provided in \cite[Proposition 11]{BEEMER202236} can be applied to $i$ transmission rounds by modeling the unambiguous code in $(\alp^{i})^3$. The lower bound provided in Proposition \ref{prop:lowA1} is a strict improvement on the lower bound provided by Proposition \ref{lwpord}. This shows that we can send more information over multiple use of the Diamond Network.

The action of the adversary on $(\mD,\ad_{\mD})$ is modeled by the channel $\h_{\mD}:\alp^3\dashrightarrow\alp^3$ defined as $\h_{\mD}(x):=\{y\in \alp^3:d_{\h}(x,y)\leq 1\}$ for all $x\in\alp^3$, where $d^\h$ represents Hamming distance. Therefore, a code $C\subseteq\alp^{3}$ is unambiguous for $\h_{\mD}$ if $d^\h(C)= 3$. We can extend this to $i$ uses of the network and describe the adversary as a power channel $\h_{\mD}^i$. Let $x=(x_1,\ldots,x_{3i})\in(\alp^3)^i$ and we can define $x^{(j)}:=(x_{3j+1},x_{3j+2},x_{3j+3})$ for all $j\in\{0,\ldots,i-1\}$. Then, for any $x\in C$, we have
\begin{multline*}
    \h_{\mD}(x)=\{y\in(\alp^{3})^i:d^\h(x^{(j)},y^{(j)}) \\\textup{ for all }j\in\{0,\ldots,i-1\}\}.
\end{multline*}
and hence $C\subseteq(\alp^i)^3$ is good code for $\h_{\mD}^i$ if and only if for all $x,y\in C$, with $x\neq y$ we have $d^\h(x^{(j)},y^{(j)})=3$ for some $j\in\{0,\ldots,i-1\}$. 

Let $s\in\{1,2,3\}$ and define $\pi_s^i:\alp^{3i}\longrightarrow\alp^i$ to be the projection onto the components in the set $\{3j+s:j\in\{0,\ldots,i-1\}\}$. Intuitively, these are the components corresponding to the edge $e_s$ in each round. The following generalizes~\cite[Claim~A]{BEEMER202236}

\begin{lemma}
Let $\mF$ be a network code for $(\mD,\ad_\mD)$. If $C\subseteq\alp^{3i}$ be an unambiguous code for $(\mD,\ad_\mD,\mF)$ then $|\pi_1^i(C)|=|C|$.
\end{lemma}
\begin{proof}
   We already showed that $C$ is unambiguous for the Diamond Network if and only if for any $x,y\in C$, with $x\neq y$, $d^\h(x^{(j)},y^{(j)})=3$ for some $j\in\{0,\ldots,i-1\}$. Suppose, towards a contradiction, that there exist $x,y\in C$ with $x\neq y$ and $\pi^i(x)=\pi^i(y)$. It implies that $x_{3j+1}=y_{3j+1}$ and therefore $d^\h(x^{(j)},y^{(j)})\leq 2$ for all $j\in\{0,\ldots,i-1\}$. This leads to a contradiction.
\end{proof}

The following result generalizes~\cite[Claim~B]{BEEMER202236}.

\begin{lemma}
Let $\mF$ be a network code for $(\mD,\ad_\mD)$. If $C\subseteq(\alp^{i})^3$ be an unambiguous code for $(\mD,\ad_\mD,\mF)$ then the restriction of $\mF_{V_1}$ to $\pi_1^i(C)$ is injective. 
\end{lemma}
\begin{proof}
    Suppose, towards a contradiction, that there exist $x,y\in C$ such that $x\neq y$ and $\mF_{V_1}(\pi_1^i(x))=\mF_{V_2}(\pi_1^i(x))$. One can check that the vector 
    \begin{equation*}
        \left(\mF_{V_1}(\pi_1^i(x))\mid \mF_{V_2}\left(\pi_2^i(x)\mid \pi_3^i(y)\right)\right)\in\alp^{2i}
    \end{equation*}
    is in $\Omega^i[\mD,\ad_\mD,\mF,\textup{out}(S) \to \textup{in}(T)](x)\cap\Omega^i[\mD,\ad_\mD,\mF,\textup{out}(S) \to \textup{in}(T)](y)$. This implies that $C$ is not unambiguous for $(\mD,\ad_{\mD})$ which is a contradiction. 
\end{proof}

Following the notation in \cite{BEEMER202236}, we let
\begin{align*}
    \overline{\Omega}^i:=&\Omega^i[\mD,\ad_\mD,\mF,\{e_1,e_2,e_3\} \to \{e_2,e_3\}],\\
    \Omega^i:=&\Omega^i[\mD,\ad_\mD,\mF,\{e_1,e_2,e_3\} \to \{e_5\}]
\end{align*}
which is well-defined since $\{e_1,e_2,e_3\}$ precedes $\{e_5\}$. The following result generalizes~\cite[Claim~C]{BEEMER202236}.

\begin{lemma}
Let $\mF$ be a network code for $(\mD,\ad_\mD)$. If $C\subseteq\alp^{3i}$ is an unambiguous code for $(\mD,\ad_\mD,\mF)$ then there exists at most one element $x\in C$ such that $|\Omega^i(x)|=1$. 
\end{lemma}
\begin{proof}
Suppose, towards a contradiction, that there exist $x,y\in C$ such that $|\Omega^i(x)|=|\Omega^i(y)|=1$ and $x \neq y$. It implies $|\mF_{V_2}(\overline{\Omega}^i(x))|=|\mF_{V_2}(\overline{\Omega}^i(y))|=1$. Since $(\pi_2^i(x)\mid\pi_3^i(x)),(\pi_2^i(x)\mid\pi_3^i(y))\in\overline{\Omega}^i(x)$ and $(\pi_2^i(y)\mid\pi_3^i(y)),(\pi_2^i(x)\mid\pi_3^i(y))\in\overline{\Omega}^i(y)$, we have 
\begin{align*}
    \mF_{V_2}(\pi_2^i(x)\mid\pi_3^i(x))&= \mF_{V_2}(\pi_2^i(x)\mid\pi_3^i(y)) \\ 
    &=\mF_{V_2}(\pi_2^i(y)\mid\pi_3^i(y)).
\end{align*}
It follows that 
\begin{multline*}
    \Omega^i[\mD,\ad_\mD,\mF,\textup{out}(S) \to \textup{in}(T)](x)\cap\\
    \Omega^i[\mD,\ad_\mD,\mF,\textup{out}(S) \to \textup{in}(T)](y)\neq\emptyset
\end{multline*}
since the adversary can corrupt the edge $e_1$. This is a contradiction.
\end{proof}

We are now ready to prove a result  analogous to \cite[Proposition~12]{BEEMER202236}.

\begin{proposition}\label{prop:upA1}
Let $\mF$ be a network code for $(\mD,\ad_\mD)$. If $C\subseteq\alp^{3i}$ is an unambiguous code for $(\mD,\ad_\mD,\mF)$ then
\begin{equation*}
    |C|^2+|C|-1-|\alp|^{2i}\leq 0.
\end{equation*}
In particular, we have $|C|\leq |\alp|^i-1$.
\end{proposition}
\begin{proof}
For ease of notation, we define
\begin{align*}
\widehat{\Omega}^i:=&\Omega^i[\mD,\ad_\mD,\mF,\{e_1,e_2,e_3\} \to \{e_4,e_5\}],\\ 
\widehat{\Omega}_1^i:=&\{y\in\widehat{\Omega}^i:\mF_{V_1}(\pi_1^i(y))=\pi_1^i(x)\},\\
    \widehat{\Omega}_2^i:=&\{y\in\widehat{\Omega}^i:\mF_{V_2}((\pi_2^i(y)\mid\pi_2^3(y)))\\
    =&(\pi_2^i(x)\mid\pi_2^3(x))\}.
\end{align*}
for any $x\in C$, and write $\widehat{\Omega}^i(x)=\widehat{\Omega}_1^i(x)\cup\widehat{\Omega}_2^i(x)$. By definition, we have $|\widehat{\Omega}^i(x)|=|\widehat{\Omega}_1^i(x)|+|\widehat{\Omega}_2^i(x)|-1$. Using the three lemmas above we get
\begin{align}\nonumber
\sum_{x\in C}|\widehat{\Omega}^i(x)|&\geq 1-2(|C|-1)+\sum_{x\in C}|C|-|C|\\\label{eq:lemmas}
&=2|C|-1+|C|^2-|C|\\\nonumber
&=|C|^2-|C|-1.
\end{align}
On the other hand, since the code is unambiguous, we have
\begin{equation*}
    \sum_{x\in C}|\widehat{\Omega}^i(x)|\leq|\alp|^{2i}.
\end{equation*}
Combining this with \eqref{eq:lemmas} we get the statement.
\end{proof}

As a consequence of Propositions~\ref{prop:lowA1} and~\ref{prop:upA1}, we have that in the Scenario~\ref{scenario1} the maximum size of an unambiguous code for $(\mD,\ad_\mD,\mF)$ is exactly $|\alp|^i-1$ and therefore $\C_i(\mD,\ad_\mD)=\frac{\log_{|\mA|}(|\mA|^i-1)}{i}$. Therefore, we have a gain in capacity over multiple uses of $(\mD,\ad_{\mD})$.

\subsection{The Mirrored Diamond Network}

The goal of this section is to explicitly compute the multishot capacity of the Mirrored Diamond Network in Scenario \ref{scenario1}.
Let $\Omega_{\mathcal{S}}:=\Omega[\mathcal{S}, \ad_{\mathcal{S}}, \mF,\textup{out}(S) \to \textup{in}(T)]$ be the channel that represents the transfer from $S$ to $T$ of $\mathcal{S}$. Assume that $t=1$. Applying the strategy provided in \cite[Proposition V.1]{BEEMER202236} one can check that $C_i(\mS,\ad_{\mS}) \geq 1$.  We will show that 
\begin{equation*}
C_i(\mS, \ad_{\mS}) = 1
\end{equation*}
and hence the Singleton Cut-Set Bound can be achieved in each transmission round independently. As in the previous section, we model the action of an adversary acting on $(\mS,\ad_{\mS})$ by the channel $\h_{\mS}: \mA^4 \dashrightarrow \mA^4$ defined by $\h_{\mS}(x) := \{y \in \alp': d^{\h}(x,y) \leq 1\}$ for all $x \in \mA$. One can check that the largest unambiguous code for $\h_{\mS}$ has cardinality $|\alp|$ and there is no unambiguous code with larger cardinality. Thus $C_1(\h_{\mS}) = 1$. 

\begin{notation} We let 
\begin{multline*}
\Omega_i' := (\Omega_{\mF_1}[\textup{out}(S) \to \textup{in}(T)] \concat \h_{\mathcal{S}}) \times \ldots\\
 \times (\Omega_{\mF_i}[\textup{out}(S) \to \textup{in}(T)] \concat \h_{\mathcal{S}})
\end{multline*}
represent the channel describing $i$ uses of the network $(\mS,\ad_{\mS})$, where $\mathcal{F}_j$ denotes the network code used in transmission round $j$, with $j\in\{1, \ldots, i\}$. Notice that $\Omega_{i}': (\mA^4)^i \dashrightarrow (\mA^2)^i$.
\end{notation}

We now show that there does not exist an unambiguous code $C\subseteq (\alp^{4})^i$ for the adversary channel $\h_{\mathcal{S}}^i$ such that $|C| = |\alp|^i + 1$.

\begin{example}\label{cne}
For $x\in(\alp^{4})^i$, we let $x^j:= (x_{4j+1},x_{4j+2},x_{4j+3},x_{4j+4})$ for all $j\in{0,\ldots,i-1}$. The code $C\subseteq(\alp^{4})^i$ is unambiguous for $\h_{\mS}^i$ if and only if one among $d^{\h}(x^1,y^1), \ldots, d^{\h}(x^i,y^i)$ has Hamming distance atleast $3$ for all $x,y$ with $x \neq y$. Let $c_1, \ldots, c_{|\mA|^i + 1}$ be elements in $C$. We want to show that there are no two codewords of $C$ that coincide in the first $|\mA|^i$ components. Assume, towards a contradiction, that $c_1^1 = c_2^1$. We can also assume that $c_1 = 0$ without loss of generality, which implies $c_2^1 = 0$. It follows that $c_3^1, \ldots, c_{|\mA|^i+1}^1$ must have Hamming weight at least $3$. If $c_3^1$ has Hamming weight less than $3,$ then one among $\{c_1^2,c_2^2,c_3^2\}, \ldots, \{c_1^i,c_2^i,c_3^i\}$ is an unambiguous code for $\h_{\mS}$ of cardinality $3$ with minimum Hamming distance $3$, which contradicts the fact that $C_1(\h_{\mS}) = 1.$ On the other hand, since $c_3^1, \ldots, c_{|\alp|^i + 1}^1$ have Hamming weight at least $3$, we have that the Hamming distance between two elements in $\{c_3^1, \ldots,c_{|\alp|^i + 1}^1\}$ is atmost $2$. Therefore assuming that $C$ is good for $\h_{\mS}^i$ implies that one of $\{c_1^2,c_2^2,c_3^2\}, \ldots, \{c_1^i,c_2^i,c_3^i\}$ must be a code with cardinality $3$ and minimum Hamming distance $3$, which again contradicts $C_1(\h_{\mS}) = 1$.
\end{example}

We are now ready to prove the main result of this section.

\begin{proposition}\label{umd} The $i$-shot capacity of the Mirrored Diamond Network in Scenario \ref{scenario1} is $$C_i(\mathcal{S}, \ad_{\mathcal{S}}) 
 = 1.$$ 
\end{proposition}
\begin{proof}
We use an approach similar to the one in \cite[Example 56]{ravagnani2018}. Let $\mF$ be a network code for $(\mathcal{S}, \ad_{\mathcal{S}})$ and assume that $\ad$ can corrupt at most $1$ edge from $\{e_1,e_2,e_3,e_4\}$. Let $x = (x_1, \ldots, x_{4i}) \in \mA^{4i}$. We have that 
\begin{multline*}
\Omega_{i}'(x) = \h_{\mS}^i(\mF^1(x_1,x_2,x_3,x_4),\ldots,\\
\mF^i(x_{4i-3},x_{4i-2},x_{4i-1},x_{4i})).
\end{multline*}
We want to show that $C_1(\Omega_{i}') < |\alp|^i + 1$. Assume that there exists a code $C \subseteq (\alp^{4})^i$ with $|C| = |\alp|^i + 1$ which is good for $\Omega_{i}'$. Then 
\begin{multline*}
C' := \{\mF^1(x_1,x_2,x_3,x_4),\ldots, \\
\mF^i(x_{4i-3},x_{4i-2},x_{4i-1},x_{4i})): x \in C\} \subseteq  \mA^{2i}
\end{multline*}
is an unambiguous code for $\h_{\mS}^i$ with $|C'| = |\mA|^i + 1$. We have that there must exist $a,a' \in C$, $a \neq a'$ such that $(a_1, \ldots, a_{|\mA|^i}) = (a_1', \ldots, a_{|\mA|^i}')$. Therefore, $C'$ is a good code for $\h_{\mS}^i$ which has two different code words that coincide in the first $|\mA|^i$ components. However, by Example~\ref{cne}, this code does not exist. This implies that $C_1(\Omega_{i}') < \textup{log}_{|\mA|}(|\mA|^i + 1)$ and we get $C_i(\mathcal{S}, \ad_{\mathcal{S}}) \leq 1$.

It remains to show that $C_i(\mathcal{S},\ad_{\mathcal{S}}) \geq 1$. This immediately follows by \cite[Proposition 12]{BEEMER202236} using the channel $\Omega^i_{\mathcal{S}}$. In particular, we have that 
\[C_i(\Omega_{\mathcal{S}}) = \frac{C_1(\Omega^i_{\mathcal{S}})}{i} \geq \frac{i C_1(\Omega_{\mathcal{S}})}{i} = C_1(\Omega_{\mS}) = 1.\]
It follows the statement.
\end{proof}

We concludes this section with the following  remark on the mulishot capacity of Family $\mathfrak{C}$ and Family $\mathfrak{D}$.

\begin{remark}\label{famR}
It is not hard to check that for an adversary able to attack up to $t$ edges, $C_i(\mathfrak{C}, \ad_{\mathfrak{C}}) = 1$ and $C_i(\mathfrak{D},\ad_{\mathfrak{D}}) = 1.$ This follows by using arguments similar to the ones in Example~\ref{cne} and in the proof of Proposition~\ref{umd}, under the assumption that the adversary can corrupt up to $t$ edges, with $t\geq 2$.
\end{remark}

In summary, we showed that in Scenario~\ref{scenario1}, there is a gain in capacity for multiple uses of the Diamond Network. In contrast  the Mirrored Diamond Network, and the networks in families $\mathfrak{C}$ and $\mathfrak{D}$, have a constant capacity $1$ over multiple uses.

\section{Multishot Regime II}
\label{sec:regII}
In this section, we study the multishot capacity of the Diamond and Mirrored Diamond Network along with the families $\mathfrak{C}$ and $\mathfrak{D}$ in Scenario \ref{scenario2}.

\subsection{The Diamond Network}
In the main result of this section, we show that the $i$-shot capacity of the Diamond Network in Scenario~\ref{scenario2} is
\begin{equation*}
\C_i(\mD,\ad_\mD)=\textup{log}_{|\alp|}{(|\alp| - 1)}.
\end{equation*}
This shows that for multiple uses of the Diamond Network, there is no gain in capacity, in contrast with Scenario \ref{scenario1}. In the following examples, we provide a characterization of an unambiguous code for $\h_{\mD}$.

\begin{example} Let $\h_{\mD}$ be the adversarial channel for $\mD$ as in Section~\ref{sec:regI}. The source $S$ can send any symbol of $\alp' = \alp \setminus \{\star\}$, where $\star$ is a reserved symbol in the alphabet $\alp$. One can check that, under the assumption of Scenario~\ref{scenario2},  the largest unambiguous code for $\h_{\mD}$ has cardinality $|\alp|-1$ and there is no larger unambiguous code. Recall that the source $S$ cannot send $\star$, implying $C_1(\h_{\mD}) = \textup{log}_{|\mA|}(|\mA|-1).$
\end{example}

The argument similar to the one in Scenario~\ref{scenario1} shows that a code $C \subseteq (\mA^{3})^i$ is unambiguous for $\h_{\mD}^i$ if and only if, for all $x,y \in C$ with $x \neq y$, we have $d^{\h}(x^{(j)}),y^{(j)}) \geq 3$, with $j\in\{1,\ldots,i-1\}$. In the next example, we show that there does not exist a code $C$ with $|C| = (|\mA| - 1)^i+1$ for $i$ uses of the Diamond Network.

\begin{example}\label{contE} 
We suppose that that there exists an unambiguous code  $C \subseteq (\mA^{3})^i$ for $\h_{\mD}^i$  such that $|C| < (|\alp|-1)^i + 1$. We use the same notation as in Example~\ref{cne}, for any $x = (x_1,\ldots,x_{3i})\in (\mA^{3})^i$, we let $x^1:= (x_1,x_2,x_3), \ldots, x^i:= (x_1^i,x_2^i,x_3^i)$ and $c_1, \ldots, c_{(|\alp|-1)^i + 1}$ be elements in $C$. We claim that there are no two codewords of $C$ that coincide in the first $(|\alp|-1)^i$ components. Assume, towards a contradiction, that $c_1^1 = c_2^1$ for $c_1^1 \neq c_2^1$ and we let $c_1 = 0$, without loss of generality. This implies that $c_2^1 = 0$ and that $c_3^1, \ldots, c_{(|\alp|-1)^i+1}^1$ must have Hamming weight at least $3$. Observe that if $c_3^1$ has Hamming weight less than $3,$ then one of $\{c_1^2,c_2^2,c_3^2\}, \ldots, \{c_1^i,c_2^i,c_3^i\}$ is a code in $\alp^3$ of cardinality $3$ with minimum Hamming distance $3$, contradicting $C_1(\h_{\mD}) = \textup{log}_{|\mA|}(|\mA|-1).$ On the other hand, since $c_3^1, \ldots, c_{(|\mA|-1)^i + 1}^1$ have Hamming weight at least $3$ which implies that the Hamming distance between any two elements in $\{c_3^1, \ldots,c_{(|\mA|-1)^i + 1}^1\}$ is at most $2$. Since we assumed that $C$ is unambiguous for $\h_{\mD}^i$, we have that one of $\{c_1^2,c_2^2,c_3^2\}, \ldots, \{c_1^i,c_2^i,c_3^i\}$ must be a code with cardinality $3$ and minimum Hamming distance $3$, which again contradicts the capacity of $\h_{\mD}$.
\end{example}

We introduce the following notation.

\begin{notation} We let 
\begin{multline*}
\Omega_i := (\Omega_{\mF_1}[\textup{out}(S) \to \textup{in}(T)] \concat \h_{\mathcal{D}}) \times \ldots \\
\times (\Omega_{\mF_i}[\textup{out}(S) \to \textup{in}(T)] \concat \h_{\mathcal{D}})
\end{multline*}
be the channel modeling $i$ uses of the network $(\mD,\ad_{\mD})$, where $\mathcal{F}_j$ denotes the network code used in transmission round $j$, with $j\in\{1, \ldots, i\}$. Note that $\Omega_{i}: (\mA^3)^i \dashrightarrow (\mA^2)^i$.
\end{notation}

We are now ready to prove the main theorem of this section.

\begin{theorem} 
The $i$-shot capacity of the Diamond Network in Scenario \ref{scenario2} is
$$\C_i(\mD,\ad_\mD)=\log_{|\alp|}(|\mA|-1).$$
\begin{proof} 

We first show that $|C| \leq (|\alp|-1)^i$. Assume that an adversary $\ad_{\mD}$ can corrupt at most of the edges in the set $\{e_1,e_2,e_3\}$. We observe that, for an $x = (x_1, \ldots, x_{3i}) \in \alp^{3i}$, we have
\begin{multline*}
    \Omega_i(x) = \h_{\mD}^i(\mF^1(x_1,x_2,x_3),\ldots,\\ \mF^i(x_{3i-2},x_{3i-1},x_{3i})).
\end{multline*}
Using a similar approach as in Proposition~\ref{umd}, we assume that there exists an unambiguous code $C \subseteq \alp^{3i}$ for $\Omega_i$ with $|C| = (|\alp|-1)^i + 1$. The code
\begin{multline*}
C' := \{\mF^1(x_1,x_2,x_3), \ldots, \\
\mF^i(x_{3i-2},x_{3i-1},x_{3i}))\} \subseteq  (\alp^{3})^i
\end{multline*}
is unambiguous for $\h_{\mD}^i$ of cardinality $(|\alp|-1)^i + 1$. Since $|C'| = (|\mA|-1)^i + 1$, there must exist $a,a' \in C$, $a \neq a'$ such that $(a_1, \ldots, a_{(|\alp|-1)^i}) = (a_1', \ldots, a_{(|\alp|-1)^i}')$. Thus $C' \in (\alp^{3})^i$ is an unambiguous code for $\h_{\mD}^i$ with two different code words that coincide in the first $(|\alp|-1)^i$ components. However such a code does not exist by Example~\ref{contE}. It follows $C_1(\Omega_i) < \textup{log}_{|\alp|}((|\alp|-1)^i + 1)$ and hence $|C| \leq (|\alp|-1)^i$. 
It remains to show that $|C| \geq (|\alp|-1)^i$. This immediately follows by applying \cite[Proposition 12]{ravagnani2018} and the strategy introduced in \cite[Proposition 11]{BEEMER202236} to the argument above.
\end{proof}
\end{theorem}

It is interesting to note that when the adversary can change the edge attacked in each transmission round, we recover the same scheme as in~\cite{BEEMER202236}. 

\subsection{Families $\mathfrak{C}$ and $\mathfrak{D}$}
In this section, we show that the multishot capacities of networks in families $\mathfrak{C}$ and $\mathfrak{D}$ is $1$. Recall that family $\mathfrak{D}$ includes the Mirrored Diamond Network.  We start with the following preliminary result.

\begin{proposition} 
The $i$-shot capacity of a network in families $\mathfrak{C}$ and $\mathfrak{D}$ in Scenario \ref{scenario2} is $1$. 
\end{proposition}
\begin{proof} The result follows similarly as in~Proposition\ref{umd} and Remark~\ref{famR}, since we do not assume anything on the adversary except that it can attack up to $1$ edge. 
\end{proof}

These results shows that, in Scenario \ref{scenario2}, there is no gain of using these networks multiple times for communication under this adversarial model.

\section{Open Questions and Future Work}
\label{sec:future}

In this paper we have investigated the multishot capacity of networks with restricted adversaries focusing on the Diamond Network and the Mirrored Diamond Network as elementary building blocks of a general theory.
In future work, we plan to investigate the multishot capacity of arbitrary networks by establishing a multishot version of the double cut-set bound of in~\cite{beemer2023network}
and by computing the multishot capacities of all the five families introduced in the same paper.

\end{document}